\documentclass[%
 reprint, 
nofootinbib,
 amsmath,amssymb,
 aps, physrev,
]{revtex4-2}
\usepackage[normalem]{ulem}
\usepackage{graphicx}
\usepackage{dcolumn}
\usepackage{bm}

\usepackage{amsmath, amssymb, amsthm}
\usepackage{physics}
\usepackage{mathtools}
\usepackage[capitalise]{cleveref}

\newtheorem{theorem}{Theorem}
\newtheorem{definition}{Definition}
\newtheorem{lemma}{Lemma}

\usepackage{tikz}
\usetikzlibrary{positioning}

\definecolor{white}{RGB}{255, 255, 255}
\definecolor{gray}{RGB}{240, 240, 240}
\definecolor{black}{RGB}{0, 0, 0}

\usepackage{color,ulem}

\begin{document}

\preprint{10.48550/arXiv.2504.08456 }

\title{Generalization Bounds in Hybrid Quantum-Classical Machine Learning Models}

\author{Tongyan Wu}
\affiliation{%
Fraunhofer Institute for Cognitive Systems IKS,  Munich, Germany}%
\affiliation{Technical University of Munich, Germany}
\author{Amine Bentellis}%
\affiliation{%
Fraunhofer Institute for Cognitive Systems IKS,  Munich, Germany}%
\email{Second.Author@institution.edu}
\author{Alona Sakhnenko}%
\affiliation{%
Fraunhofer Institute for Cognitive Systems IKS,  Munich, Germany}%
\affiliation{Technical University of Munich, Germany}
\author{Jeanette Miriam Lorenz}%
\affiliation{%
Fraunhofer Institute for Cognitive Systems IKS,  Munich, Germany}%
\affiliation{Ludwig-Maximilian University, Munich, Germany}


\date{\today}

\begin{abstract}
Hybrid classical-quantum models aim to harness the strengths of both quantum computing and classical machine learning, but their practical potential remains poorly understood. In this work, we develop a unified mathematical framework for analyzing generalization in hybrid models, offering insight into how these systems learn from data. We establish a novel generalization bound of the form  $\tilde{\mathcal O}\left(
\tfrac{\alpha^{k}}{\sqrt{N}}\,
\big( k^{\tfrac{3}{2}}\sqrt{m n}\;+\;\sqrt{T\log T}\big)
\right)$ for \( N \) training data points, \( T \) trainable quantum gates,  \( n \) dimensional quantum circuit output,
and $k$ bounded linear layers \( \|F_i\|_F \leq \alpha\) where \( i = 1, \dots, k \) and $F_i \in \mathbb{R}^{m \cross n}$ interspersed with activation functions. This generalization bound decomposes into quantum and classical contributions, providing a theoretical framework to separate their influence and clarifying their interaction. Alongside the bound, we highlight conceptual limitations of applying classical statistical learning theory in the hybrid setting and suggest promising directions for future theoretical work.

\end{abstract}

\maketitle


\section{Introduction}
Quantum-classical hybrid machine learning models have emerged as a promising approach in an era where fully fault-tolerant, large-scale quantum computers remain out of reach. Currently available devices impose significant restriction on width and depth of quantum circuits, making hybridization of quantum and classical computation an ideal candidate for exploiting early capabilities of quantum computing. In the current framework, classical optimization methods act as a complement for shallow-depth circuits. As the quantum technological stack matures into fault-tolerant regime, hybrid algorithms will evolve as well, allowing a careful interleaving of classical and quantum, leveraging the strength of both of these resources. 
Early proof-of-concept hybrid algorithms show promising results \cite{Caro2022, Abbas_Power, Huang_2021, Zoufal_2019,  K_lis_2023, Sakhnenko_2022, qccnn, Gili_2023}, however, a solid theoretical understanding of the power of these architectures is lacking. Gaining such insights could shed light onto the differences in learning capacity  between quantum and classical algorithms.

There are several theoretically interesting aspects of a learning model’s behavior, e.g. resources demand \cite{Caro2022}, sampling complexity \cite{Gili_2023} and learnability \cite{Sweke_2021}. One of the most sought-after metrics is the generalization, which allows to predict its performance beyond the available dataset. Classical approaches in statistical learning theory  establish uniform upper bounds on the generalization error based on both the training set size and the inherent complexity of the model. In this context, complexity can be characterized through measures, such as covering numbers and Rademacher complexity \cite{NIPS2017_b22b257a}, which capture the richness of a function class associated with the hypothesis class \footnote{A concept that represents a set of all possible functions that are representable by a chosen algorithm.} of the model. Several works have provided rigorous bounds on the generalization error for classical machine learning models \cite{he2021recentadvancesdeeplearning, Abbas2021, Zhang2021} as well as for fully quantum learning models \cite{Caro2022, Abbas_Power, Huang_2021}. However, despite the growing popularity of hybrid models, the conditions under which these hybrid architectures can generalize accurately remain largely unexplored. 


Inspired by the remarkable success of neural networks (NNs) in classical machine learning, the quantum machine learning community has pursued a similar line of research. Layer-wise organized quantum circuits with parametrizable gates that mimic classical feedforward architecture are known as quantum NNs (QNNs). A diverse zoo of models has emerged inspired by classical ideas, such as hybrid Bolzmann machines \cite{K_lis_2023}, hybrid autoencoders \cite{Sakhnenko_2022} and hybrid convolutional NNs \cite{qccnn}. QNNs fall under the category of quantum machine learning models (QMLM) \cite{Biamonte2017, Schuld2015}, which in general consists of a parametrized quantum circuit that is being trained on classical or quantum data.

In this paper, we derive and prove a generalization bound for hybrid (quantum-classical) machine learning models (hybrid QMLMs) that consist of general QMLMs combined with classical NNs.
In the hybrid quantum-classical setting, a model's hypothesis class becomes richer due to integrating quantum with classical machine learning components. In our proof, we first separately derive complexity measures for the hypothesis classes corresponding to the quantum and classical components. Due to the richness of the hybrid model's hypothesis class, we use covering numbers to quantify the complexity of the hypothesis class geometrically, rather than a direct worst-case analysis using Rademacher complexity. After establishing covering number bounds, we combine them to characterize the overall complexity of the hybrid hypothesis class. Using Dudley’s entropy integral, we obtain generalization bounds that quantify model's learning ability from finite data. A key implication of our results is that the complexity of a hybrid model that has minimal classical layers will perform closely to a full QMLM, providing theoretical support for the practical viability and design of quantum-assisted learning algorithms. As such, designing hybrid models with shallow classical components but expressive quantum circuits can achieve strong generalization while reducing implementation overhead, offering a principled framework for balancing quantum and classical resources in real-world QMLM architectures.

This paper is organized as follows: In Section~\ref{sec:related}, we justify the chosen QMLM architecture and provide an overview of previous work on generalization bounds in both quantum and classical contexts. In Section~\ref{sec:background}, we outline the fundamental concepts relevant to this work. In Section~\ref{sec:results}, we present a derived generalization bound for a hybrid quantum-classical setup and discuss its implications in Section~\ref{sec:discussion}. Finally, Section~\ref{sec:generalization_bound} contains a step-by-step derivation of the aforementioned bound.
\section{Related work}\label{sec:related}

On the quantum side, \citet{Caro2022} prove generalization bounds for QMLMs and derive a bound for (full) quantum CNN (QCNN) based on covering numbers with $T$ trainable gates to scale at worst as $\Tilde{\mathcal{O}}\left(\sqrt{\frac{T \operatorname{log}T}{N}} \right)$. Furthermore, they show that when re-running the same gates at most $M$ times, the generalization performance scales $\Tilde{\mathcal{O}}\left(\sqrt{\frac{T \operatorname{log}MT}{N}} \right)$, only worsening logarithmically. These findings suggest that these models can generalize well even with a small training dataset. In our work, we extend these results to capture the behaviour of a hybrid architecture.

On the hybrid side, while the theoretical foundation is currently lacking (providing motivation for this work) numerous studies have demonstrated promising empirical results. For instance, \citet{qccnn} presented empirical evidence supporting the viability of a hybrid quantum-classical CNNs (QCCNNs) in a medical use case, demonstrating its potential practical utility in real-world scenarios. A subsequent study~\cite{10821430} studied theoretically derived generalization metrics for QMLMs, such as \cite{Abbas_Power, Schuld_2021}, applied to architecture from \cite{qccnn} in an empirical context;  however, these efforts revealed a significant lack of correlation between  theoretical metrics and empirical accuracy. This implies that the general metrics did not capture the behaviour of a hybrid model sufficiently well. In contrast, in this work, we take a fully theoretical approach and develop generalization bounds tailored to a hybrid setup.

On the classical side, extensive research has been conducted on generalization bounds for NNs. \citet{lin2019} offers a comprehensive overview of various approaches to these bounds and provides a simplification and generalization of them under specific assumptions. Notably, the bounds $\mathcal{O}\left( \frac{2^k \eta^k}{\sqrt{N}} \right)$ \cite{Neyshabur2015} and $\tilde{\mathcal{O}}\left( \frac{\alpha^{k - 1} k^{\frac{3}{2}} \eta \sqrt{d} }{\sqrt{N}} \right)$ \cite{NIPS2017_b22b257a, Neyshabur2017}, where $||F||_{\sigma} \leq \eta$, emphasize that generalization depends on how weight magnitudes (spectral norm of a weight matrix) scale with depth and the activation functions are assumed to be 1-Lipschitz. The bounds illustrate how generalization bounds can be established via norm-based complexity measures, particularly the product of per-layer norms and depth. These bounds are reflected in our analysis of the classical component of hybrid QMLMs. An important point to highlight is that our work prioritizes the investigation into interaction between the classical and quantum components rather than the sharpness of the bounds. Future research can focus on refining these bounds, as our study establishes a robust foundation that facilitates this process.

\section{\label{sec:background}Background}

\subsection{Generalization}
Generalization is the ability of a (trained) model to perform well on unseen data. We measure the generalization error by comparing the model's performance on training data with its expected performance on unseen data. Formally, given a training dataset $S = \{(x_i ,y_i)\}, i=1, \dots N $ of size $N$ and the model's hypothesis class $H$, the empirical risk associated with a hypothesis $h \in H$ is given by

\begin{equation} \hat{R}_S(h) = \frac{1}{N} \sum_{i=1}^{N} \ell(h(x_i), y_i), \end{equation}
where $\ell$ is a predefined loss function that quantifies the discrepancy between the model prediction and the true label. This empirical risk, also known as the training error, serves as an estimate of the expected risk, which is defined as

\begin{equation} R(h) = \mathbb{E}_{(x,y) \sim \mathcal{D}} [\ell(h(x), y)], \end{equation}
where the expectation is taken over the unknown data distribution $\mathcal{D}$.  The discrepancy between these two quantities, known as the generalization error,

\begin{equation} \operatorname{gen}(h) = R(h) - \hat{R}_S(h), \end{equation}
measures how well the training performance translates to unseen data. Bounding this generalization error is crucial in statistical learning theory and is the focus of the subsequent theoretical results, where maintaining good generalization to new data while fitting the training data well is a trade-off.
\subsection{Complexity measures}
Complexity measures quantify how well functions from a hypothesis class are able to fit various training data, as overfitting can be understood as a class having many functions that are distinct from the target function but can still fit the training data well. 
The Rademacher complexity and covering numbers are complexity measures based on probability and geometric characterizations of hypothesis classes and have well-established connections to translate between them. Our approach builds on existing uniform-bound based measures in the quantum setting, where these have proven effective in bounding QMLM's \cite{Caro2022}.
\subsubsection{Rademacher complexity}
\begin{definition}[Empirical Rademacher Complexity]\label{def:EmpRC}
 Consider arbitrary spaces $\mathcal{X}, \mathcal{Y}$, define $\mathcal{Z}:=\mathcal{X} \times \mathcal{Y}$.
  Given a training set $S = \{z_1, \dots, z_n\}$ and a real-valued hypothesis class $H$ on $\mathcal{Z}$, the empirical Rademacher complexity is defined as the expected supremum of the correlation between the functions in $H$ and the random Rademacher variables $\sigma_i$ applied to the dataset:

\begin{equation}
\hat{\mathcal{R}}_S(H) \coloneq \mathbb{E}_\sigma \left[ \sup_{h \in H} \frac{1}{n} \sum_{i=1}^n \sigma_i h(z_i) \right].
\end{equation}

 where the Rademacher variables $\sigma_i \sim U_{\{-1,1\}}$ are random variables that take values $\pm 1$ with equal probability. 
\end{definition}



From the definition, we see that empirical Rademacher complexity varies with the training set, reflecting the actual learning difficulty of each dataset, explaining its name. This data-dependence can yield tighter generalization bounds than distribution-independent worst-case bounds, yet it can still be used to derive bounds that hold for all distributions. As such, Rademacher complexity often serves as a key step in proving general results in learning theory. The following theorem is fundamental in this context.

\begin{theorem}[Rademacher Generalization Bounds\cite{ABartlett1999}]\label{thm:Radbound} Consider arbitrary spaces $\mathcal{X}, \mathcal{Y}$. We define $\mathcal{Z}:=\mathcal{X} \times \mathcal{Y}$ and a real-valued hypothesis class $H$ on $\mathcal{Z}$. For any $\delta>0$ and any probability measure $\mathbb{P}$ on $\mathcal{Z}$ we have with probability at least $(1-\delta)$ for the training data $S \in \mathcal{Z}^n$ obtained by $n$-times repeated sampling w.r.t. $\mathbb{P}$:
 \begin{equation}  \operatorname{gen}(h) \leq 2 \hat{\mathcal{R}}_S(\ell\circ H)+c \sqrt{\frac{\log \frac{1}{\delta}}{n}}, \quad \forall h \in H,
 \end{equation} where $c > 0$ is some constant and $\ell\circ H:=\{(x, y) \mapsto \ell(h, y, x) \mid h \in H\}$ is the loss transformed hypothesis class. 
\end{theorem}
This theorem allows us to bound the generalization error using Rademacher complexities.
\subsubsection{Covering Number}
We begin by defining the $\varepsilon$-cover and the covering number.

\begin{definition}( $\varepsilon$-cover and covering number \cite{Wolf2023}) Let $(V,\|\cdot\|)$ be a normed space, and let $U \subseteq V$. Then $U$ is an $\varepsilon$-cover of $\,V$ if $\, \forall v \in V$, there exists $u \in U$ such that $\|u-v\| \leq \varepsilon$. The covering number of the normed space $(V,\|\cdot\|)$ with any $\varepsilon>0$ is the size of the smallest $\varepsilon$-cover. \begin{equation}
  \mathcal{N}(V, \|\cdot\|, \varepsilon) \coloneq \min \{|U|: U \text{ is an } \varepsilon\text{-cover of } V\}. \end{equation}
We call the logarithm transformed covering number $\log(\mathcal{N}(V, \|\cdot\|, \varepsilon) )$ the metric entropy of $V$.\end{definition} 

The larger the covering number, the more complex the hypothesis class, because it indicates that more functions are needed to approximate the entire class within a small error tolerance. Intuitively, if a hypothesis class $H$ has a large covering number for small $\varepsilon$, it means the class is complex, as there are many distinct functions that cannot be closely approximated by one another. On the other hand, if the covering number is small, the class is simpler, with many hypotheses being similar to each other within the margin of error $\varepsilon$.
There is a common argument chain to arrive at generalization bounds via Thm.~\ref{thm:Radbound} and an entropy integral (developed by R.M. Dudley\cite{dudley}) once a covering number has been determined for a hypothesis class.
\begin{theorem}[Dudley's Entropy Integral \cite{Shalev2014}]\label{thm:Dudley}
 
 Let $H$ be a hypothesis class on $\mathcal{X}$ equipped with a norm $\norm{\cdot}$ and let $\gamma_0:=\sup _{h \in H}\|h\|$. We can then bound the empirical Rademacher complexity by Dudley's entropy integral: The usage of 
 
 $$
 \hat{\mathcal{R}}_S(H) \leq \underset{ \alpha \in [0, \frac{\gamma_0}{2})}{\operatorname{inf}} 4\alpha +  \frac{12}{\sqrt{n}} \int_\alpha^{\gamma_0} \sqrt{\log \mathcal{N}( H, \|\cdot\|,\varepsilon)} d \varepsilon.
 $$
\end{theorem}
The integral can be understood as summing the contributions of each infinitesimal resolution to the Rademacher complexity of the hypothesis class. At each level of resolution, the function class will become better at fitting the Rademacher random variables.

\section{Results}\label{sec:results}
With this work, we lay the foundation for more theoretically grounded understanding of the capability of hybrid quantum-classical models. We derive and prove a generalization bound $
\tilde{\mathcal O}\!\left(
\frac{\alpha^{k}}{\sqrt{N}}\,
\big( k^{\tfrac{3}{2}}\sqrt{m n}\;+\;\sqrt{T\log T}\big)
\right)$ for \( N \) training data points, \( T \) trainable quantum gates, \( n \) dimensional quantum circuit output,
and bounded fully-connected $k$ layers \( \|F_i\|_F \leq \alpha\) where \( i = 1 \dots k, F \in \mathbb{R}^{m\cross n} \). The bound can be decomposed into classical and quantum components, which enables investigation of their interaction. Just like the bounds presented by \citet{lin2019} on the classical side, we assumed 1-Lipschitz acivations functions when presenting our bound, however the derivation in Appendix~\ref{sec:proof} includes the general case for $L$-Lipschitz activation functions.

Our analysis leverages concepts from statistical learning theory and quantum information, allowing us to theoretically predict the empirical performance of hybrid models. The derived generalization bound indicates how the model's complexity, dictated by the number of quantum gates and the characteristics of the classical layers, influences its ability to generalize from the training data to unseen examples.
A key implication of our result is that introducing classical layers into the architecture does not introduce significant generalization disadvantage compared to a fully quantum model with equivalent quantum capacity. In fact, the bound suggests that a well-designed hybrid model with a moderate number of classical layers can match the generalization capabilities of a purely quantum model while offering practical advantages such as reduced quantum circuit depth and better noise tolerance. 
On the other side, classical NNs typically exhibit generalization bounds that scale unfavorably with the number of layers, often exponentially. In contrast, the hybrid model's quantum component contributes a polylogarithmic factor, potentially offering more favorable scaling under certain conditions.
This provides a theoretical justification for hybridization as a strategy for improving scalability without sacrificing learning performance. 
\section{Discussion}
\label{sec:discussion}



While the theoretical bounds investigated in this work provide valuable insights, they have been shown to struggle in the overparameterized regime. Recent works, such as \cite{Zhang2021} and its quantum variant \cite{Eisert2024}, argue that traditional generalization bounds fall short of explaining the empirical success despite high complexities of modern machine learning models. To bridge the gap between theoretical bounds and empirical observations, future research should explore new paradigms of generalization that synergize current understanding of generalization based on model complexities with emerging evidence from edge cases. This approach will foster a more comprehensive understanding of generalization in diverse scenarios. Recent approach \cite{Canatar_2021} has demonstrated strong empirical support for theoretically established bounds in an overparameterized regime. The authors leveraged the relationship between kernel methods and deep learning via Neural Tangent Kernels, validating the idea presented in \cite{belkin2018understanddeeplearningneed} that comprehending generalization in deep learning necessitates understanding generalization in kernel methods. This suggest that future research into generalization bounds for NNs can benefit from a synergistic investigation into their more foundational cousin - the kernel methods. Similar connections between models has been identified in quantum case~\cite{ QuantumTangent, Incudini_2023}, suggesting that this line of research might shed light on generalization bounds of QNNs as well.

The classical part of the bound presented in our work suffers from the same problems mentioned above. While these bounds align with earlier covering-number results, such as those in \cite{ABartlett1999}, they are looser than the more recent bounds found in \cite{Neyshabur2015} and \cite{Neyshabur2017}. However, our primary contribution is the separation of the classical and quantum parts by their covering numbers. The submultiplicative property is applicable to any chosen covering number for each component under appropriate norms. Our approach prioritizes a methodological perspective over the tightness offered by modern results.

Apart from understanding the generalization abilities of quantum models, it is crucial to identify the optimal context (use-case) for utilizing these models. Recent research  \cite{bermejo2024quantumconvolutionalneuralnetworks} has shown that QCNNs can be efficiently simulated on classical devices though a process known as dequantization. This is particularly efficient on easy dataset that are usually used to demonstrate the utility of proof-of-concept implementations. These results can be extended to a broader QNN architecture. In pursuit of identifying a ``killer application" for QMLMs, it is essential to evaluate both the architecture and its generalization capabilities, alongside the complexity of the specific use case at hand. By integrating the tools developed in the present work with insight from \cite{bermejo2024quantumconvolutionalneuralnetworks}, we bring us one step closer to understand generalization benefits of hybrid QMLMs.

Our study has some limitations, which should be addressed in the future work. Firstly, we focused on a purely theoretical derivation of the bound, and empirical validation remains to be conducted. However, the work \cite{Caro2022} which is foundational for our study has already performed empirical verification of their bound, suggesting that our results should exhibit similar behavior. Secondly, we consider a hybrid architecture comprising a quantum layer followed by classical layers. While the derived generalization bound provides insight into the learning behavior of the hybrid quantum-classical model, it does not directly answer how to determine the optimal quantum-to-classsical ratio of layers. The bound primarily consists of the insight provided by \cite{Caro2022} where the generalization bound on the QMLM scales less than exponentially in the training data and the generalization bounds of NNs. However, since both these bounds are loose \cite{Eisert2024, Zhang2021}, they fail to capture the interplay between large circuits and deep classical layers.

\section{Methods}\label{sec:generalization_bound}
In this section, we present the theoretical methods that lead to the generalization bounds for QMLMs by combining bounds from quantum (see Section~\ref{sec:quantum_part}) and classical (see Section~\ref{sec:classical_part}) parts as illustrated in \cref{fig:proof_structure}. For a hybrid QMLM and a sufficiently large training data set, our bounds guarantee accurate generalization that is close to the performance of a fully QMLM with high probability on unseen data.

\begin{figure*}
    \begin{tikzpicture}

    \node[draw, draw opacity=0, rotate=90, text opacity=0.3] at (-0.5,4.9) {\textit{Classical}};
    \node[draw, draw opacity=0, rotate=90, text opacity=0.3] at (-0.5,3.05) {\textit{Quantum}};
    
    \draw [thick,dashed, draw](-0.75,4) -- (12,4);

    \node[draw, draw opacity=0, align=center, minimum width=4.4cm] at (2,5.75) {\small \cref{lem:metricnn}\\ \textit{Metric entropy for}\\[-0.2em] \textit{bounds for $k$-layer NN}};
    \node[draw, draw opacity=0, align=center, minimum width=4.1cm, fill=white] at (2,4) {\small \cref{lem:mainresult}\\ \textit{Submultiplicativity of}\\[-0.2em]\textit{hybrid covering numbers}};
    \node[draw, draw opacity=0, align=center, minimum width=4.4cm] at (2,2.25) {\small \cref{lem:covunit}\\ \textit{Metric entropy for}\\[-0.2em] \textit{bounds for QMLM}};

    \node[draw, draw opacity=0, align=center] at (5.4,3.2) {\small \cref{thm:Dudley}\\ \textit{Dudley's entropy}\\[-0.2em]\textit{integral}};

    \node[draw, draw opacity=0, align=center] at (5.4,4.5) {\small \cref{thm:Radbound} \\ \textit{Rademacher Bound}};

    \node[draw, draw opacity=0, align=center, minimum width=4.1cm, fill=gray, minimum height=2.5cm] at (9,4) {\small \cref{thm:main}\\ \textbf{Generalization bounds}\\[-0.2em]\textbf{for hybrid QMLM}};
 
    \draw [-stealth, draw](2,2.8) -- (2,3.4);
    \draw [-stealth, draw](2,5.2) -- (2,4.6);
    \draw [-stealth, draw](4,3.9) -- (6.9,3.9);

    \end{tikzpicture}
    \caption{Visualization of the proof structure: We present metric entropy/covering number bounds for both quantum and classical components, demonstrating that the hybrid metric entropy bound can be expressed through both. We then utilize this framework to derive our generalization bound via Dudley's entropy integral and Rademacher bounds.}
    \label{fig:proof_structure}

\end{figure*}

\subsection{Classical Part: Covering Numbers in Neural Networks}
\label{sec:classical_part}

For the classical part, we introduce the metric entropy bound following the proofs in \cite{ABartlett1999}, which derived generalization bounds for NNs using covering numbers.

Intuitively, a covering number represents the smallest number of covers needed, based on a chosen radius or resolution, to cover a given space. In the context of function spaces, this concept characterizes how uniformly any two functions can vary from each other. When considering bounded linear functions with a bound $\alpha$, we can visualize the space as being enclosed within a sphere of radius $\alpha$. Lemma~\ref{lem:metricblf} extends this intuition by deriving the covering number, or more precisely, the metric entropy associated with this space. Realizing that a classical NN consists of linear layers and activation functions, we can write down the hypothesis class of a NN with $k-$layers and extend Lemma ~\ref{lem:entropyfc} to derive a NN's metric entropy. The idea is to cover the changes in each layer through coverings of each linear layer while dealing with the interspersed activation function. However, separating out the $L$-Lipschitz activation functions will pick up a factor $L$ each layer together with the operator bound. The change to the uniform output metric reflects the datum dependent formulation of Thm.~\ref{thm:Dudley}, since we are formulating the generalization bound over possible data input $x$. The proof can be found in the Appendix ~\ref{sec:proof}.

\begin{lemma}[Metric Entropy for classical $k$-Layer Neural Networks]\label{lem:metricnn}
Let $\mathcal X = \{x \in \mathbb{R}^n : \|x\|_{\ell_2} \le R\}$ be the input domain. 
Let $\sigma:\mathbb{R}\to\mathbb{R}$ be an $L$-Lipschitz activation function applied coordinate-wise, and consider the class of $k$-layer networks
$
\mathcal{C} = \{f(x)=F_k\,\sigma\!\big(F_{k-1}\,\sigma(\cdots \sigma(F_1 x)\big),
\; \|F_i\|_F \le \alpha,\; F_i \in \mathbb{R}^{m \times n} \}.
$
Then, for any $\varepsilon > 0$, the covering number of $\mathcal{C}$ under the data dependent 
$\ell_2$ metric
$$
\|f-g\|_{X, \ell_2}
\;:=\;
\left( \frac{1}{n} \sum_{j=1}^n \|f(x_j)-g(x_j)\|_2^2 \right)^{1/2}
$$
satisfies
$$
\log \mathcal N\!\left(\mathcal{C}, \|\cdot\|_{X, \ell_2}, \varepsilon\right)
\;\le\;
k\,m n \cdot \log\!\left(\tfrac{3\,k\,R\,L^{\,k-1}\,\alpha^{\,k}}{\varepsilon}\right),
$$
for $\varepsilon \le k\,R\,L^{\,k-1}\alpha^k$, and equals $0$ for larger $\varepsilon$.
\end{lemma}

With the metric entropy for the $k$-layer NN, we can derive the generalization bound of the classical part. Inserting Lemma~\ref{lem:metricnn} into Thm.~\ref{thm:Dudley} we obtain a bound on the empirical Rademacher complexity of the hypothesis class. This induces the generalization bound using Thm.~\ref{thm:Radbound}. Hence,

\begin{theorem}[Generalization Bound for $k$-layer Neural Networks]
\label{thm:genNN} 
Let $\sigma:\mathbb{R}\to\mathbb{R}$ be an $L$-Lipschitz activation applied coordinatewise, and consider the class of $k$-layer networks
$
\mathcal{C} = \{f(x) = F_k\,\sigma\!\big(F_{k-1}\,\sigma(\cdots \sigma(F_1 x)\big),
\; \|F_i\|_F \le \alpha,\; F_i \in \mathbb{R}^{m \times n} \}.
$

Suppose the loss $\ell$ is coordinate-wise $L$-Lipschitz and bounded by $M$.
Then with probability at least $1-\delta$ over the choice of i.i.d. training data $S$ of size $N$ according to $\mathbb{P}$,

\begin{equation}
|R(f) - \hat{R}_S(f)| \; \in \mathcal{O}\left(\, L^{k}\,\alpha^{k}\, k^{\tfrac{3}{2}} \,\sqrt{\tfrac{mn}{N}} + \sqrt{\tfrac{\log (1 / \delta)}{N}}\right).
\end{equation}

\end{theorem}

\subsection{Quantum Part: Generalization Bound in QMLMs}
\label{sec:quantum_part}
We introduce the generalization bound for QMLM derived from the covering number of 2-qubit quantum channels and unitary gates, which is formalized in Lemma~\ref{lem:covunit}. This lemma provides an upper bound for the covering number of the set of 2-qubit unitaries by using the fact that the space of 2-qubit unitaries is bounded within a unit ball under the operator norm.  

\begin{lemma}(Covering number bounds for 2-qubit unitaries \cite{Caro2022}).\label{lem:covunit} Let $\norm{\cdot}$ be a unitarily invariant norm on complex $4 \times 4$-matrices. The covering number of the set of 2-qubit unitaries $\mathcal{U}\left(\mathbb{C}^2 \otimes \mathbb{C}^2\right)$ w.r.t. the norm $\norm{\cdot}$ can be bounded as
	\begin{equation}
	\mathcal{N}\left(\mathcal{U}\left(\mathbb{C}^2 \otimes \mathbb{C}^2\right),\|\cdot\|, \varepsilon\right) \leqslant\left(\frac{6\left\|I_{\mathbb{C}^4}\right\|}{\varepsilon}\right)^{32},
	\end{equation}

$0<\varepsilon \leqslant\left\|I_{\mathbb{C}^4}\right\|$.
\end{lemma}
Remark: The metric entropy of $2$-qubit unitaries can then easily be seen as $$ \log\mathcal{N}\left(\mathcal{U}\left(\mathbb{C}^2 \otimes \mathbb{C}^2\right),\|\cdot\|, \varepsilon\right) \leqslant 32 \log \left(\frac{6\left\|I_{\mathbb{C}^4}\right\|}{\varepsilon}\right). $$

\begin{theorem}[Generalization Bound for QMLM \cite{Caro2022}]
For a QMLM $\mathcal{E}_\theta^{\mathrm{QMLM}}$ using $T$ parametrized local quantum channels, we have with high probability over training data of size $N$ that 
\begin{equation}
| R ( \mathcal{E}_\theta^{\mathrm{QMLM}} )-\hat{R} ( \mathcal{E}_\theta^{\mathrm{QMLM}} ) | \in \mathcal{O}\left(\sqrt{\frac{T \log(T)}{N}}\right).
\end{equation}
\end{theorem}

\textbf{Remark}: This implies that the required size of the data $N$ scales as
$$
\sqrt{\frac{T \log(T)}{N}} < \epsilon \implies N > \frac{T \log(T)}{\epsilon^2}
$$
A more refined version of this statement is the following theorem that considers a QMLM where many of the parameterized (local) gates are applied repeatedly. Assume each gate is repeated at most $M$ times.

\begin{theorem}[Generalization Bound for Repeated Local Gates \cite{Caro2022}]
\label{thm:gencaro}
Let $\mathcal{E}_\theta^{\mathrm{QMLM}}$ be a QMLM with an architecture consisting of $T$ independently parameterized 2-qubit Completely Positive Trace-Preserving (CPTP) maps and at most repeated usage of these channels $M$ times. Then, with probability at least $1-\delta$ over the choice of i.i.d. training data $S$ of size $N$ according to $\mathbb{P}$,
\begin{equation}
\begin{aligned}
R\left(\mathcal{E}_\theta^{\mathrm{QMLM}}\right)&-\hat{R}_S\left(\mathcal{E}_\theta^{\mathrm{QMLM}}\right) \\
&\in \mathcal{O}\left(\sqrt{\tfrac{T \log (T M)}{N}}+\sqrt{\tfrac{\log (1 / \delta)}{N}}\right).
\end{aligned}
\end{equation}

\end{theorem}

\subsection{Generalization Bound in the Hybrid QMLM}
\label{sec:Hybrid proof}

Now that we have derived the covering numbers for each part of the hybrid quantum machine learning model, we investigate how these layers interact to establish the generalization bound. 

A hybrid QNN integrates both classical and quantum computational parts. Mathematically, this can be viewed as a hypothesis class that consists of functions mapping classical data to quantum states, then back to a classical result via measurement. 
Formally, we denote a QMLM as 

$$
\mathcal{E}^{\mathrm{QMLM}}_{\theta, x}(\cdot), 
$$

where $\theta$ are continuous parameters of the quantum gates, i.e. angles in rotation channels and $x$ is the classical input datum. 

We assume in the hybrid setting that the classical datum $x$ is encoded into a quantum state with density matrix $\rho(x)$. As the specific encoding technique has no impact on our proof, we will sometimes omit its dependence on $x$ and just write $\rho$. The classical NN receives the expected measurement outcome $\operatorname{tr}\left(M \cdot \mathcal{E}^{\mathrm{QMLM}}_{\theta, x}\left( \rho(x) \right)\right)$, where $M$ is a measurement operator. We also assume that the classical NN follows empirical risk minimization (ERM) and the optimization over quantum parameters is performed using a classical optimizer, typically gradient-based methods. The model, therefore, consists of a composite hypothesis class
\begin{equation}\label{eqn:hybhypothesis} 
H = \left\{ h(x) = F \cdot \operatorname{tr}\left(M\cdot \mathcal{E}_{{\theta}, x}^{\mathrm{QMLM}} \right) \,\Big|\, F \in \mathcal{C} \right\},
\end{equation}
where the quantum component \( \mathcal{E}_{\boldsymbol{\theta}, x}^{\mathrm{QMLM}} \) represents a parametrized quantum channel applied to the encoded classical input \( x \), followed by a measurement operator \( M \). 
The classical component \( F \in \mathcal{C} \) is a classical NN with $k$-layers of linear operators with bounded norm and $L$-Lipschitz activation functions as defined in Lemma ~\ref{lem:metricnn}.

\begin{theorem}[Generalization Bound in Quantum-Classical Hybrid Models]\label{thm:main}

Let $h_\theta$ be a Hybrid-QMLM with an architecture consisting of $T$ independently parameterized 2-qubit Completely Positive Trace-Preserving (CPTP) maps followed by bounded measurement operators $\norm{M} \leq \beta$ and a $k$-layer classical NN consisting of activation functions and bounded linear layers $\norm{F}_F \leq \alpha$. Then, with probability at least $1-\delta$ over the choice of i.i.d. training data $X=\{x_i\}_{i=1}^N$ of size $N$ according to $\mathbb{P}$,
\begin{equation}
R\left(h_\theta \right)-\hat{R}_X\left(h_\theta \right) \in \tilde{\mathcal O}\!\left(
\frac{\alpha^{k}}{\sqrt{N}}\,
\big( k^{\tfrac{3}{2}}\sqrt{m n}\;+\;\sqrt{T\log T}\big)
\right)
\end{equation}

where $\tilde{\mathcal{O}}(n)$ denotes Big-O notation with poly-logarithmic terms in $n$ hidden.
\end{theorem}

The proof consists of combining the two results Lemma~\ref{lem:metricblf} and Lemma~ \ref{lem:covunit} to obtain a covering number for the hybrid model, after which we establish the generalization bound. Essentially, we demonstrate that the covering numbers exhibit a submultiplicative property across each layer by addressing how the classical and quantum parts interact as separate layers. This then leads to a standard argument of bounding the Rademacher complexity via Dudley's entropy integral which then delivers the generalization bound.

We combine previous work on both classical and quantum machine learning models to claim a similar result for hybrid models with multiple classical layers. For a hybrid QMLM and a sufficiently large training data set, a slightly adjusted proof gives a bound that guarantees accurate generalization close to the performance of a fully QMLM with high probability on unseen data.

\textbf{Remark}: For a single layer this bound reduces to the quantum part asymptotically 

$$
\tilde{\mathcal{O}}\big( \frac{\alpha}{\sqrt{N}} \sqrt{m n}\;+\;\sqrt{T\log T}\big) = \tilde{\mathcal{O}}\Big(\sqrt{\frac{T \log(T)}{N}}\Big)
$$and implies that the generalization error of a hybrid model with a single classical layer is not much worse than a fully quantum model.

\section{Conclusion}
This work contributes a theoretical foundation for understanding generalization in hybrid quantum-classical machine learning models. We provide the first characterization of learning capacity in hybrid models by decomposing into distinct quantum and classical contributions. Our result shows that introducing bounded classical layers on top of a trainable quantum model does not degrade generalization performance. Instead, hybrid architectures can retain the expressivity and learning guarantees of their fully quantum counterparts while offering practical benefits, such as reduced quantum circuit depth and improved robustness to noise. 
While our bound advances the theoretical understanding of hybrid models, it leaves open the important question of how to determine the optimal balance between quantum and classical components. Addressing this challenge will require new theoretical tools or empirical studies that go beyond the current statistical learning theory tools. Nonetheless, our work provides a stepping stone for principled design and evaluation of hybrid models and opens several promising directions for future research in quantum machine learning theory.
\section{Acknowledgements}
This research is supported by the Bavarian Ministry of Economic Affairs, Regional Development and Energy with funds from the Hightech Agenda Bayern. 

\bibliography{apssamp}

\appendix 
\section{}
\label{sec:proof}
In this section, we present the proof for the claimed generalization bounds for hybrid QMLMs by combining bounds from the quantum \ref{sec:quantum_part} and the classical \ref{sec:classical_part} components. The proof is based on showing the submultiplicativity property of the covering numbers when combining both parts.

The first part introduces two lemmata that lead to the classical result from \cite{ABartlett1999} on generalization bounds for NN by deriving a metric entropy bound on the NN. The presented Lemma~\ref{lem:metricblf} on bounded linear functionals has been adapted from~\cite{Wolf2023}. By characterizing the interaction of each layer in the NN, since the layers are bounded linear functionals interspersed with activation functions, we can derive a metric entropy bound for the entire NN.

\begin{lemma}[$\ell_2$ Metric Entropy for Bounded Linear Functionals]
\label{lem:metricblf}
Consider the set $\mathcal{V} = \{ v \in \mathbb{R}^n \mid \norm{v}_{\ell_2} \leq \alpha \}$ of bounded linear functionals (by identifying the vectors $v$ with their duals and interpreting them as functionals on $\mathbb{R}^n$). Then, for any $\varepsilon > 0$, we have
\begin{equation}
\log\mathcal{N}(\mathcal{V}, \norm{\cdot}_{\ell_2}, \varepsilon) \leq n \cdot \log\left(\frac{3\alpha}{\varepsilon}\right).
\end{equation}
\end{lemma}

\begin{proof}
Notice that $\mathcal{V} = \{v \in \mathbb{R}^n \mid \norm{v}_{\ell_2} \leq \alpha\}$, which is a $\ell_2$-ball of radius $\alpha$. The result follows from plugging into the metric entropy for norm-balls.
\end{proof}

\begin{lemma}[Metric Entropy for Fully Connected Layer]\label{lem:entropyfc}
Let $\mathcal{F} = \{F \in \mathbb{R}^{m \times n} \mid \norm{F}_F \leq \alpha \}$, where $m$ is the output dimension and $n$ the input dimension. This set represents bounded matrices presenting a fully connected layer in a neural network. Its metric entropy bound is given by 
\begin{equation}
\log\mathcal{N}(\mathcal{F}, \norm{\cdot}_{F}, \varepsilon) \leq nm \cdot \log\left(\frac{3\alpha}{\varepsilon}\right).
\end{equation}
\end{lemma}

\begin{proof}
Recall that $\mathbb{R}^{m \times n}$ is isomorphic to $\mathbb{R}^{mn}$ and notice that $\norm{F}_F = \norm{\phi(F)}_{\ell_2}$, where $\phi$ is the isomorphism between the two spaces. Thus, we apply Lemma~\ref{lem:metricblf} and derive the result.
\end{proof}

\begin{lemma}[Metric Entropy for classical $k$-Layer Neural Networks]\label{lem:metric_entropy_NN}
Let $\mathcal X = \{x \in \mathbb{R}^n : \|x\|_{\ell_2} \le R\}$ be the input domain. 
Let $\sigma:\mathbb{R}\to\mathbb{R}$ be an $L$-Lipschitz activation function applied coordinatewise, and consider the class of $k$-layer networks
$
\mathcal{C} = \{f(x)=F_k\,\sigma\!\big(F_{k-1}\,\sigma(\cdots \sigma(F_1 x)\big),
\; \|F_i\|_F \le \alpha,\; F_i \in \mathbb{R}^{m \times n} \}.
$
Then, for any $\varepsilon > 0$, the covering number of $\mathcal{C}$ under the data dependent 
$\ell_2$ metric
$$
\|f-g\|_{X,\ell_2}
\;:=\;
\left( \frac{1}{n} \sum_{j=1}^n \|f(x_j)-g(x_j)\|_2^2 \right)^{1/2}
$$
satisfies
$$
\log \mathcal N\!\left(\mathcal{C}, \|\cdot\|_{X, \ell_2}, \varepsilon\right)
\;\le\;
k\,m n \cdot \log\!\left(\tfrac{3\,k\,R\,L^{\,k-1}\,\alpha^{\,k}}{\varepsilon}\right),
$$
for $\varepsilon \le k\,R\,L^{\,k-1}\alpha^k$, and equals $0$ for larger $\varepsilon$.
\end{lemma}

\begin{proof}
Write the network recursively as
\begin{equation}
\begin{aligned}
&f_0(x):=x, \; \; \quad a_i(x):=\sigma\!\big(f_i(x)\big), \\
&f_i(x):=F_i\,a_{i-1}(x)\qquad (i=1,\dots,k).
\end{aligned}
\end{equation}

By induction using $\|F_i\|_F\le\alpha$ and Lipschitzness of $\sigma$,
\begin{equation}\label{eq:growth}
\begin{aligned}
&\|a_{i-1}(x)\|_{\ell_2} \;\le\; (L\alpha)^{\,i-1}\,\|x\|_{\ell_2} \le\; (L\alpha)^{\,i-1}\,R, \;\\
&\forall x\in\mathcal X,\ \ i=1,\dots,k.
\end{aligned}
\end{equation}

Let $F=f_k(x)$ and $K=h_k(x)$ be two $k$-layer NN. We want to bound $\norm{f_k - h_k}_{X, \ell_2}$.
Replace layers one at a time to obtain the standard telescoping decomposition and using $\|(F_i-K_i)v\|_{\ell_2}\le \|F_i-K_i\|_F\,\|v\|_{\ell_2}$ and \eqref{eq:growth},
\begin{align*}
\|f_k(x)-h_k(x)\|_{\ell_2}
&\le \sum_{i=1}^k (L\alpha)^{\,k-i}\,\|F_i-K_i\|_F\ \|a_{i-1}(x)\|_{\ell_2} \\
&\le \sum_{i=1}^k (L\alpha)^{\,k-i}\,\|F_i-K_i\|_F\ \cdot (L\alpha)^{\,i-1}R \\
&= R\,(L\alpha)^{\,k-1}\,\sum_{i=1}^k \|F_i-K_i\|_F,
\; \forall x\in\mathcal X.
\end{align*}
Since this bound is uniform in $x$, it implies the same inequality in the empirical metric:
\begin{equation}\label{eq:emp-metric}
\|f_k(x)-h_k(x)\|_{X,\ell_2}
\ \le\ R\,(L\alpha)^{\,k-1}\,\sum_{i=1}^k \|F_i-K_i\|_F.
\end{equation}

Choose a per-layer Frobenius radius
$$
\delta \;:=\; \frac{\varepsilon}{k\,R\,(L\alpha)^{\,k-1}}.
$$
If each layer $F_i$ is approximated by $K_i$ within $\delta$ in Frobenius norm, then \eqref{eq:emp-metric}
gives $\|f_k(x)-h_k(x)\|_{X,\ell_2}\le \varepsilon$.
Thus an $\varepsilon$-cover of $\mathcal C$ is obtained by taking the product of
$\delta$-covers of each layer’s parameter ball $\{F_i:\|F_i\|_F\le \alpha\}$.

By Lemma~\ref{lem:entropyfc}, taking products over $k$ layers and
logarithms,
\begin{equation}
\begin{aligned}
&\log \mathcal N\!\left(\mathcal C,\|\cdot\|_{X,\ell_2},\varepsilon\right)\\
\ &\le\ \sum_{i=1}^k mn \log\!\Big(\tfrac{3\alpha}{\delta}\Big)\\
\ &=\ k\,mn \log\!\left(\frac{3\alpha\,k\,R\,(L\alpha)^{\,k-1}}{\varepsilon}\right)\\
\ &=\ k\,mn \log\!\left(\tfrac{3\,k\,R\,L^{\,k-1}\,\alpha^{\,k}}{\varepsilon}\right).
\end{aligned}
\end{equation}

\end{proof}

\begin{theorem}[Generalization Bound in Quantum-Classical Hybrid Models]
\label{thm:main_full}

Let \( H \) be a hypothesis class for a hybrid machine learning model with $T$ quantum gates and a $k$-layer NN, and $X=\{x_i\}_{i=1}^N$ be a sample of size $N$. Assume that the covering number of \( H \) with respect to \( \|\cdot\|_{X, \ell_2} \) satisfies, for any \( \varepsilon > 0 \),
\begin{equation}
\begin{aligned}
			&\mathcal{N}\left(H,\|\cdot\|_{X, \ell_2}, \varepsilon\right) \leq \mathcal{N}\left(\mathcal{F},\|\cdot\|_{X, \ell_2},\frac{\varepsilon}{2\beta\sqrt{n}}\right) \cdot \\ &\mathcal{N}\left(\mathcal{U}\left(\mathbb{C}^2 \otimes \mathbb{C}^2\right),\|\cdot\|,  \frac{\varepsilon}{4TL^{k-1}\alpha^k \beta \sqrt{n}}\right)^T .
		\end{aligned}
\end{equation}
where $\mathcal{U}\left(\mathbb{C}^2 \otimes \mathbb{C}^2\right)$ denotes the set of 2-qubit unitary operators, $\mathcal{F}$ denotes a $k$-layer NN with linear layers as defined in Lemma~\ref{lem:metricnn}. $M$ is a fixed measurement operator with operator norm $\norm{M} \leq \beta$, $T$ is the number of parameterized unitaries, and $n$ is the number of output registers of the quantum circuit.

Furthermore, assume that the loss function \( \ell: H \times \mathcal{X} \times \mathcal{Y} \to [0, M] \) is coordinatewise \( L \)-Lipschitz continuous in its first argument with respect to the norm \( \|\cdot\| \).
Then, with probability at least \( 1 - \delta \) over the random draw of a sample \( X \) with size $N$, for all hypotheses \( h \in H \), the expected loss satisfies:
\begin{equation}
\operatorname{gen}(h) \in \tilde{\mathcal O}\big(\frac{L^{\,k}\alpha^{\,k}}{\sqrt{N}}
( k^{\tfrac{3}{2}}\sqrt{m n}+\sqrt{T\log T})\big),
\end{equation}
where $\tilde{\mathcal{O}}(n)$ denotes Big-O notation with poly-logarithmic terms in $n$ hidden.
\end{theorem}
\begin{proof}
The proof will be split over the next subsections. We first derive the covering number of the hypothesis class $H$ of the hybrid QMLM. This is done by realizing that any hypothesis can be covered by coverings of the classical and the quantum component. After deriving the covering number, we will calculate the Rademacher complexity of the hypothesis class via Dudley's entropy integral. Finally, we bound the Rademacher complexity of the loss transformed hypothesis class by scaling the Rademacher complexity of the hypothesis class by a Lipschitz constant that is averaged over all coordinates $L$. This delivers the generalization bound.

\subsection*{Submultiplicativity of Covering Numbers for a Hybrid Quantum-Classical Model}
The hybrid model under consideration comprises two main components: a quantum component and a classical component. The essential step to finding a covering number for the hybrid model is correctly identifying how a covering can be constructed from coverings of the classical and quantum components. Lem.~\ref{lem:covunit} and Lem.~\ref{lem:metricblf} each provide covering numbers and we will show that multiplying both provides an upper bound on the hybrid model's covering number. This submultiplicativity property is the following lemma.	
\begin{lemma}\label{lem:mainresult}
		Let $\mathcal{U}\left(\mathbb{C}^2 \otimes \mathbb{C}^2\right)$ denote the set of $2$-qubit unitary operators, let $n$ be the number of output registers in the quantum circuit and let $\mathcal{F}$ denote a $k$-layer NN with Frobenius norm $\norm{F_i}_F \leq \alpha$ for each linear layer $F_i$ and $L$-Lipschitz activation functions. Let $M$ be a fixed measurement operator with operator norm $\norm{M} \leq \beta$. Define the hypothesis class $H$ as introduced in Equation~\ref{eqn:hybhypothesis} with $T$ the number of parametrized unitaries.
		Let $\mathcal{N}\left(\mathcal{U}\left(\mathbb{C}^2 \otimes \mathbb{C}^2\right),\|\cdot\|, \varepsilon\right)$ and $\mathcal{N}\left(\mathcal{F},\|\cdot\|_F, \varepsilon\right)$ be the covering numbers for the 2-qubit unitary operators and the $k$-layer NN derived previously and let $X=\{x_i\}_{i=1}^N$ be a sample of size $N$.
		
		Then, for any $\epsilon > 0$, the covering number of $H$ with respect to $\norm{\cdot}_{X, \ell_2}$ satisfies
		\begin{align*}
			&\mathcal{N}\left(H,\|\cdot\|_{X, \ell_2}, \varepsilon\right) \leq \mathcal{N}\left(\mathcal{F},\|\cdot\|_{X, \ell_2},\frac{\varepsilon}{2\beta\sqrt{n}}\right) \cdot \\ &\mathcal{N}\left(\mathcal{U}\left(\mathbb{C}^2 \otimes \mathbb{C}^2\right),\|\cdot\|,  \frac{\varepsilon}{4TL^{k-1}\alpha^k \beta \sqrt{n}}\right)^T.
		\end{align*}
\end{lemma}	

\begin{proof}

We will build coverings for (i) each unitary layer in operator norm $\|\cdot\|$, and
(ii) the classical NN $F$ in the sample dependent semi-norm $\|\cdot\|_{X, \ell_2}$.
The overall cover of $H$ will be the corresponding product covering.

We aim to construct an $\varepsilon$-covering for the hypothesis class $H$ by combining $\varepsilon_F$-coverings for the classical component $\mathcal{F}$ and $\varepsilon_U$-coverings for the quantum component $\mathcal{U}$. The key idea is to ensure that perturbations in $F$ and $U$ individually lead to a controlled perturbation in the composite hypothesis $h$.
Let $\mathcal{C}_{\mathcal{F}}(\varepsilon_F)$ be an $\varepsilon_F$-covering of the fully connected layers $\mathcal{F}$ with respect to $\norm{\cdot}_F$, and $\mathcal{C}_{\mathcal{U}}(\varepsilon_U)$ be an $\varepsilon_U$-covering of the unitaries $\mathcal{U}$ with respect to the spectral norm (induced 2-norm) $\norm{\cdot}_{\ell_2}$. We determine $\varepsilon_F$ and $\varepsilon_U$ accordingly, such that a hypothesis in $H$ is covered by its components w.r.t. $\norm{\cdot}_{\ell_2}$.
We will write a hypothesis $h$ as
$$
h(x)=F \cdot z(x),
\;
z(x) = \operatorname{tr}\left[M\,\mathcal{E}^{\mathrm{QMLM}}_{\theta,x}(\rho(x))\right]_{i=1}^n \in\mathbb{R}^n,
$$
where $\rho(x)$ is the density operator of the encoded data $x$ and $M$ is a measurement operator\footnote{For complete mathematical accuracy, we should specify that the general form of the local 2-qubit unitary operators consists of tensor products $I \otimes \dots \otimes U_k \otimes \dots \otimes I$.}. This separates the classical part $F \in \mathcal{C}$ and the quantum part $z(x)$. As a reminder, $$\mathcal{E}^{\mathrm{QMLM}}_{\theta} = U_T\dots U(\cdot)U^\dagger \dots U_T^\dagger,
$$ where the $U_i \in \mathcal{U}$ are 2-qubit unitaries depending on $\theta_i$.

Consider two hypotheses $h, h' \in H$. As introduced above, we can write them as
\begin{equation}
\begin{aligned}
    &h(x) = F \cdot z(x)\\
    &h'(x) = K \cdot z'(x),
\end{aligned}
\end{equation}
where $F, K \in \mathcal{F}$ are the classical parts and $z, z'$ are the quantum parts. 

To construct a covering for the hypothesis class $H$, we want to bound 
\begin{equation}\label{eqn:AB}
\norm{h - h'}_{X, \ell_2}.
\end{equation} 

We will bound this by separating the classical layers and the quantum layer and then bound them individually. The separation uses the common insertion trick in triangle inequalities:
\begin{equation}
h(x)-h'(x)
= \big(Fz(x)-Fz'(x)\big) + (F-K)\,z'(x).
\end{equation}

Inserting this into Eqn.~\ref{eqn:AB} gives us
\begin{equation}\label{eqn:separated}
\begin{aligned}
&\|h-h'\|_{X, \ell_2} = \|\big(Fz-Fz'\big) +(F-K)\,z' \|_{X, \ell_2} \\
& \leq  \|\big(Fz-Fz'\big) \|_{X, \ell_2} + \| (F-K)\,z' \|_{X, \ell_2}.
\end{aligned}
\end{equation}
The Eqn.~\ref{eqn:separated} shows the varied parts that need to be covered. We will (i) first bound the term $ \|\big(Fz-Fz'\big) \|_{X, \ell_2} $ corresponding to the varied quantum part and (ii) then the second term $\| (F-K)\,z' \|_{X, \ell_2}$ corresponding to the varied classical NN.

Starting with the quantum part, notice that for any $z,z'\in\mathbb R^n$,
\begin{align*}
\|F(z)-F(z')\|_{\ell_2}
&= \big\|F_k\big[\sigma(\cdots)-\sigma(\cdots)\big]\big\|_{\ell_2} \\
&\le \|F_k\| \ \big\|\sigma(\cdots)-\sigma(\cdots)\big\|_{\ell_2} \\
&\le \|F_k\|_F \ \big\|\sigma(\cdots)-\sigma(\cdots)\big\|_{\ell_2} \\
&\le \alpha \ \big\|\sigma(\cdots)-\sigma(\cdots)\big\|_{\ell_2} \\
&\le \alpha\ L\ \big\|F_{k-1}\sigma(\cdots)-F_{k-1}\sigma(\cdots)\big\|_{\ell_2}\\
&\ \ \vdots \\
&\le L^{\,k-1}\,\alpha^{\,k}\ \|z-z'\|_{\ell_2}.
\end{align*}
Thus, 
$$\|F(z)-F(z')\|_{X,\ell_2}
\ \le\ L^{\,k-1}\,\alpha^{\,k}\ \|z-z'\|_{X,\ell_2}.
$$

We now examine the varied quantum parts $\|z-z'\|_{\ell_2}$. Notice that for a measurement operator $M$ with $\|M\| \le \beta$, and $\Phi,\Psi$ quantum channels, and a state $\rho$ with $\|\rho\|_1=1$, we have
\begin{equation}\label{eq:meas-diamond}
\begin{alignedat}{2}
&|z_j(x)-z'_j(x)|
= \big|\operatorname{tr}\!\big(M\,(\Phi-\Psi)(\rho(x))\big)\big|& \\[4pt]
&\le \|M\|\,\big\|(\Phi-\Psi)(\rho(x))\big\|_1 \; \; \, \,  (\text{H\"older's Inequality (Schatten)})& \\[4pt]
&\le \|M\|\,\|\Phi-\Psi\|_\diamond \qquad \qquad \qquad \qquad \; \,  (\text{Def. Diamond norm})& \\[4pt]
&\le \beta\,\|\Phi-\Psi\|_\diamond. \qquad \qquad \qquad \qquad \qquad \qquad \quad \; \, \; \, (\|M\|\le \beta)&
\end{alignedat}
\end{equation}

Consequently, if we define the $n$-dimensional measurement vector
$$
z(x) \coloneqq \left[\operatorname{tr}\!\big(M\,\mathcal{E}_{\boldsymbol{\theta},x}^{\mathrm{QMLM}}(\rho(x))\big)\right]_{j=1,\dots,n} \in\mathbb{R}^n 
$$
then for any two channels $\mathcal{E},\mathcal{E}'$,
\begin{equation}\label{eq:z-l2-diamond}
\|z(x)-z'(x)\|_{\ell_2}
\;\le\; \beta\,\sqrt{n}\;\big\|\mathcal{E}_{\boldsymbol{\theta},x}-\mathcal{E}'_{\boldsymbol{\theta},x}{}\big\|_\diamond .
\end{equation}

Using the same approach as in the proof of Lemma C.1 from \cite{Caro2022}, let
$\mathcal U \coloneqq \mathcal{E}_{\boldsymbol{\theta},x}$ and $ \mathcal V \coloneqq \mathcal{E}'_{\boldsymbol{\theta},x}$ be the unitary channels of the quantum parts $z$ and $z'$. 
A telescoping expansion gives
\begin{align}
\big\|\mathcal U-\mathcal V\big\|_\diamond
\;\le\; 2\sum_{k=1}^T \|U_k - V_k\|.
\label{eq:depthT-diamond}
\end{align}

Combining Eqn.~\ref{eq:z-l2-diamond} and Eqn.~\ref{eq:depthT-diamond} gives 
\begin{equation}
    \|z-z'\|_{X, \ell_2}
\;\le\; 2 \beta\,\sqrt{n}\; \sum_{k=1}^T \|U_k - V_k\| .
\end{equation}
Finally, for the quantum part we get
\begin{equation}\label{eqn:quantumpart} \tag{i}
    \|Fz-Fz'\|_{X, \ell_2}
\;\le\; 2 L^{\,k-1}\,\alpha^{\,k}\beta\,\sqrt{n}\; \sum_{t=1}^T \|U_t - V_t\| .
\end{equation}
Moving on to the classical part, notice that $$\| (F-K)\,z' \|_{X, \ell_2}$$ fulfills the conditions for Lemma~\ref{lem:metricnn} with $R = \beta \sqrt{n}$, since $\norm{z(x)}_{\ell_2} \le \beta \sqrt{n}$. This let's us know the covering number for any covering of radius $\epsilon$ for the $k$-layer NN $F$. Together, we see that 
\begin{equation}
\begin{aligned}    
\|h-h'\|_{X, \ell_2}  \leq \;  &2 L^{\,k-1}\,\alpha^{\,k} \beta\,\sqrt{n}\; \sum_{k=1}^T \|U_k - V_k\| \\
&+ \| (F-K)\,z' \|_{X, \ell_2}.
\end{aligned}
\end{equation}
Set
\begin{equation}
\begin{aligned}
&\|U_k - V_k\| \leq \frac{\epsilon}{2 \cdot 2T L^{\,k-1} \alpha^k \beta \sqrt{n}} = \varepsilon_U, \\
&\|(F - K)z'\|_{X, \ell_2} \leq \frac{\epsilon}{2\beta \sqrt{n}} = \varepsilon_F.
\end{aligned}
\end{equation}
Then, for any $h \in H$, there exists a covering \(\mathcal{C}_{H}\) with $h' \in \mathcal{C}_{H}$ since
\begin{equation}
\norm{h - h'}_{X, \ell_2} \leq  \frac{\varepsilon}{2} + \frac{\varepsilon}{2} = \varepsilon.
\end{equation}
Hence, the constructed cover \(\mathcal{C}_{H}\) satisfies:
\[
\sup_{h \in H} \inf_{h' \in \mathcal{C}_{H}} \|h - h'\|_{X, \ell_2} \leq \epsilon,
\]
and since each hypothesis in $H$ can be approximated within $\varepsilon$ by combining an $\varepsilon_F$-cover for $\mathcal{F}$ and $T$ many $\varepsilon_U$-cover for $\mathcal{U}$, the total number of elements in $\mathcal{C}_{H}(\varepsilon)$ is at most the product of the covering numbers of $\mathcal{F}$ and $\mathcal{U}$ $T$-times, we proved that:
\begin{equation}\label{eq:hybridcover}
\begin{aligned}
			&\mathcal{N}\left(H,\|\cdot\|_{X, \ell_2}, \varepsilon\right) \leq \mathcal{N}\left(\mathcal{F},\|\cdot\|_{X, \ell_2},\frac{\varepsilon}{2\beta\sqrt{n}}\right) \cdot \\ &\mathcal{N}\left(\mathcal{U}\left(\mathbb{C}^2 \otimes \mathbb{C}^2\right),\|\cdot\|,  \frac{\varepsilon}{4TL^{k-1}\alpha^k \beta \sqrt{n}}\right)^T.
\end{aligned}
\end{equation}
Applying the logarithm and using the product rule gives the metric entropy $H$ by
\begin{equation}\label{eq:metricentropyH}
\begin{aligned}
&\log\mathcal{N}\left(H,\|\cdot\|_{\ell_2}, \varepsilon\right) \leq \\ &\log\mathcal{N}\left(\mathcal{U}\left(\mathbb{C}^2 \otimes \mathbb{C}^2\right),\|\cdot\|,  \frac{\varepsilon}{4T \alpha \beta \sqrt{n}}\right) \cdot T \;+ \\
&\log\mathcal{N}\left(\mathcal{F},\|\cdot\|_F,\frac{\varepsilon}{2 \beta \sqrt{n}}\right).
\end{aligned}
\end{equation}
\end{proof}

\subsubsection*{Rademacher Complexity}

\begin{lemma}[Radius of the hybrid hypothesis class]\label{lem:hybrid-radius}
Let $H$ be the hybrid hypothesis class from Lemma~\ref{lem:mainresult} producing
$n$-dimensional outputs via a measurement operator $M$ with $\|M\|\le \beta$.
Define $\gamma_0:=\sup_{h\in H}\|h\|_{X,\ell_2}$, where
$\|h\|_{X,\ell_2}^2 := \frac{1}{N}\sum_{i=1}^N \|h(x_i)\|_{\ell_2}^2$.
Then
\[
\gamma_0 \;\le\; \beta \sqrt{n}.
\]
\end{lemma}

\begin{proof}
Each coordinate of $h(x)$ is a measurement expectation of the form
$\operatorname{tr} (M\rho)$ with $| \operatorname{tr} (M\rho) |\le \|M\|\,\|\rho\|_1\le \beta$,
independently of the classical layers. Hence $\|h(x)\|_{\ell_2}\le \beta\sqrt{n}$ for every input $x$,
and therefore
$\|h\|_{X,\ell_2}^2=\frac{1}{N}\sum_{i=1}^N \|h(x_i)\|_{\ell_2}^2 \le \beta^2 n$.
Taking the supremum over $h\in H$ gives $\gamma_0\le \beta\sqrt{n}$.
\end{proof}

\begin{theorem}[Rademacher Complexity Bound for the Hybrid Class]\label{thm:dudley-hybrid}
Under the setting of Lemma~\ref{lem:mainresult}, with the neural-network
metric-entropy lemma and the diamond-norm covering for the quantum part, we have
for all $\varepsilon>0$,
\begin{equation}
\begin{aligned}
\log \mathcal N\!&\left(H,\|\cdot\|_{X,\ell_2},\varepsilon\right)
\;\le\log \mathcal{N}\left(\mathcal{F},\|\cdot\|_{X, \ell_2},\frac{\varepsilon}{2\beta\sqrt{n}}\right)\\
+&T\log \mathcal N\left(\mathcal{U}\left(\mathbb{C}^2 \otimes \mathbb{C}^2\right),\|\cdot\|,\tfrac{\varepsilon}{4TL^{k-1}\alpha^{k}\beta\sqrt{n}}\right).
\end{aligned}
\end{equation}
Consequently, Dudley’s entropy integral (with cutoff) yields
$$
\hat{\mathcal R}_N(H) \in 
\tilde{\mathcal{O}}\left(\frac{L^{\,k-1}\alpha^{\,k}\beta\sqrt{n}} {\sqrt{N}}\, (kR\sqrt{kmn} + \sqrt{T\log T})\right).
$$

\end{theorem}

\begin{proof}
We start by inserting the metric entropy derived in Lemma\eqref{lem:mainresult} into Dudley's Theorem (Thm.~\ref{thm:Dudley}):
\begin{equation}
\hat{\mathcal{R}}_N(H) \leq \frac{12}{\sqrt{N}}\int_{0}^{\gamma_0}\sqrt{\log \mathcal{N}\left(\mathcal{F}\right) + T\log \mathcal N \left(\mathcal{U}\right)}d\varepsilon.
\end{equation}
We omitted the metric and resolution size for readability. By Lemma~\ref{lem:hybrid-radius}, $\gamma_0$ represents the upper bound of the integration bounds, since it is the coarsest resolution for the hypothesis class $H$.
Then by applying the inequality \( \sqrt{a + b} \leq \sqrt{a} + \sqrt{b} \):
\begin{align}
&\hat{\mathcal R}_N(H) \le \frac{12}{\sqrt{N}}\bigg( \int_{0}^{\gamma_0}\sqrt{\log \mathcal{N}\left(\mathcal{F},\|\cdot\|_{X, \ell_2},\frac{\varepsilon}{2\beta\sqrt{n}}\right)}d\varepsilon\\
&+ \int_{0}^{\gamma_0}\sqrt{T\log \mathcal N \left(\mathcal{U}\left(\mathbb{C}^2 \otimes \mathbb{C}^2,\|\cdot\|,\tfrac{\varepsilon}{4TL^{k-1}\alpha^{k}\beta\sqrt{n}}\right) \right)}d\varepsilon\bigg),
\end{align}

We analyze each summand separately:
Starting from the classical contribution, insert Lemma~\ref{lem:metric_entropy_NN}. The integral becomes directly evaluatable,
\begin{equation}\label{eqn:classical_integral}
\frac{12}{\sqrt N}\int_{0}^{\gamma_0}
\sqrt{kmn\log\!\Big(\tfrac{6kRL^{k-1}\alpha^{k}\,\beta\sqrt n}{\varepsilon}\Big)}\,d\varepsilon
=\frac{12\sqrt{k\,m\,n}}{\sqrt N}\;J,
\end{equation}
where
$$
J:=\int_{0}^{\gamma_0}\sqrt{\log\!\Big(\tfrac{C}{\varepsilon}\Big)}\,d\varepsilon,\qquad
C:=6kR\,L^{\,k-1}\alpha^{\,k}\,\beta\sqrt n.
$$
Evaluate $J$ by the change of variables $u=\log(C/\varepsilon)$, i.e.\ $\varepsilon=C\,e^{-u}$ and
$d\varepsilon=-C\,e^{-u}\,du$. The limits become
$\varepsilon\downarrow 0 \Rightarrow u\to\infty$ and $\varepsilon=\gamma_0 \Rightarrow u=\log(C/\gamma_0)$, hence
$$
J=C\int_{\log(C/\gamma_0)}^{\infty} u^{1/2}e^{-u}\,du
= C\,\Gamma\!\Big(\tfrac{3}{2},\,\log\tfrac{C}{\gamma_0}\Big),
$$
where $\Gamma(s,x)$ is the gamma function. Using the identity
$\Gamma(\tfrac{3}{2},x)=\tfrac{\sqrt{\pi}}{2}\,\operatorname{erf}(\sqrt{x})+\sqrt{x}\,e^{-x}$ and $e^{-\log(C/\gamma_0)}=\gamma_0/C$,
we obtain
$$
J=\frac{\sqrt{\pi}}{2}\,C\,\operatorname{erf}\!\Big(\sqrt{\log\tfrac{C}{\gamma_0}}\Big)
\;+\;
\gamma_0\,\sqrt{\log\tfrac{C}{\gamma_0}},
$$
where $\operatorname{erf}(x) \coloneqq \tfrac{2}{\sqrt{\pi}}\int_0^x \exp(-t^2) \;dt $. Multiplying $J$ by $\frac{12\sqrt{k\,m\,n}}{\sqrt N}$ gives us

$$
\frac{12\sqrt{k\,m\,n}}{\sqrt N} \left(\frac{\sqrt{\pi}}{2}\,C\,\operatorname{erf}\!\Big(\sqrt{\log\tfrac{C}{\gamma_0}}\Big)
\;+\;
\gamma_0\,\sqrt{\log\tfrac{C}{\gamma_0}} \right),
$$
which is the classical contribution \ref{eqn:classical_integral} we started with. 

Using the fact that $\operatorname{erf}(\cdot) \leq 1$ and hiding constants and poly-logarithmic terms using the $\tilde{\mathcal{O}}$ notation, we see that the Rademacher complexity of the classical part is
$$
\tilde{\mathcal{O}}\left(\sqrt{kmn}\frac{kRL^{\,k-1}\alpha^{\,k}\beta\sqrt{n}} {\sqrt{N}}\right).
$$ 

Moving on to the quantum part, we see
\begin{equation}
\frac{12}{\sqrt{N}} \int_{0}^{\gamma_0} \!\!\!\!\sqrt{
T \log \mathcal N\left(\mathcal U\left(\mathbb C^{2}\otimes\mathbb C^{2}, \|\cdot\|,\ 
\tfrac{\varepsilon}{4TL^{k-1}\alpha^{k}\beta\sqrt{n}}\right)\right)
} d\varepsilon.
\end{equation}
With the change of variables
\[
\delta \;=\; \frac{\varepsilon}{4\,T\,L^{\,k-1}\alpha^{\,k}\,\beta\sqrt{n}},\qquad
b \;=\; \frac{\gamma_0}{4\,T\,L^{\,k-1}\alpha^{\,k}\,\beta\sqrt{n}},
\]
this becomes
\begin{equation}\label{eq:q-dudley-change}
\frac{48TL^{k-1}\alpha^{\,k}\beta\sqrt{n}}{\sqrt{N}}
\int_{0}^{b}
\sqrt{ T \log \mathcal N\left(\mathcal U(\mathbb C^{2}\otimes\mathbb C^{2}),\ \|\cdot\|,\ \delta\right)}\; d\delta.
\end{equation}
By Lemma~\ref{lem:covunit}, for any unitarily invariant norm $\|\cdot\|$ and $0<\delta\le \|I_{\mathbb{C}^4}\|$,
\[
\log \mathcal N\!\left(\mathcal U(\mathbb C^{2}\!\otimes\!\mathbb C^{2}),\ \|\cdot\|,\ \delta\right)
\;\le\; 32\,\log\!\Big(\tfrac{6\|I_{\mathbb{C}^4}\|}{\delta}\Big).
\]
Plugging this into \eqref{eq:q-dudley-change} and using that for the spectral norm $\|I_{\mathbb{C}^4}\|=1$, yields
\begin{equation}\label{eq:q-dudley-preclosed}
\frac{48\,T\,L^{\,k-1}\alpha^{\,k}\,\beta\sqrt{n}}{\sqrt{N}}\;\sqrt{32\,T}\;
\int_{0}^{b} \sqrt{\,\log\!\Big(\tfrac{6}{\delta}\Big)}\; d\delta.
\end{equation}
Evaluating the integral exactly gives
\[
\int_{0}^{b} \sqrt{\,\log\!\Big(\tfrac{6}{\delta}\Big)}\; d\delta
\;=\;
\frac{\sqrt{\pi}}{2}\,6\,
\operatorname{erf}\!\Big(\sqrt{\log\tfrac{6}{b}}\Big)
\;+\;
b\,\sqrt{\log\tfrac{6}{b}}.
\]
Hence, when inserting back \( b=\dfrac{\gamma_0}{4\,T\,L^{\,k-1}\alpha^{\,k}\,\beta\sqrt{n}} \), the quantum Dudley term becomes
\begin{equation}\label{eq:q-dudley-closed}
\frac{3\sqrt{512}}{\sqrt{N}}\;\gamma_0\,\sqrt{T}\;
\sqrt{\log\!\Big(\frac{24T\,L^{\,k-1}\alpha^{\,k}\,\beta\sqrt{n}}{\gamma_0}\Big)}.
\end{equation} Using $\gamma_0\le L^{\,k-1}\alpha^{\,k}\beta\sqrt{n}$ we see that the quantum part scales as \(\displaystyle \frac{L^{\,k-1}\alpha^{\,k}\beta\sqrt{n}}{\sqrt{N}}\,\sqrt{T\log T}\) up to explicit constants.

Combining both expressions, we arrive at the claimed

$$
\tilde{\mathcal{O}}\left(\frac{L^{\,k-1}\alpha^{\,k}\beta\sqrt{n}} {\sqrt{N}}\, (kR\sqrt{kmn} + \sqrt{T\log T})\right).
$$

\end{proof}

\subsubsection*{Generalization Bound}

Having obtained an upperbound on the empirical Rademacher complexity \(\hat{\mathcal{R}}_N(H)\) of the hypothesis class \(H\), we aim to apply Theorem~\ref{thm:Radbound} to establish a corresponding generalization bound. However, it is important to observe that Theorem~\ref{thm:Radbound} requires the empirical Rademacher complexity of the loss-transformed hypothesis class \(\mathcal{R}_S(\ell \circ H)\), rather \(\hat{\mathcal{R}}_S(H)\). To overcome this gap, we will introduce the following lemma which provides an upper bound on \(\mathcal{R}(\ell \circ H)\) in terms of \(\mathcal{R}(H)\) for loss functions that are Lipschitz-continuous, such as the multi-class hinge loss. We will state it without proof.

\begin{lemma}[Talagrand-Ledoux's Contraction Lemma\cite{ledoux}]
Let \( H \) be a hypothesis class and let \( \ell: H \times \mathcal{Y} \to \mathbb{R} \) be an \( L \)-Lipschitz loss function with respect to its first argument. For any fixed sample \( S = \{x_1, x_2, \dots, x_n\} \), the $L$-Lipschitz loss will scale $\frac{L}{\sqrt{n}}$ when averaged over the sample. Thus the empirical Rademacher complexity of the loss-transformed hypothesis class \( \ell \circ H \) satisfies:
\begin{equation}
\hat{\mathcal{R}}_S(\ell \circ H) \leq \frac{L}{\sqrt{n}} \cdot \hat{\mathcal{R}}_S(H)
\end{equation}
\end{lemma}
By applying this lemma, we can estimate the Rademacher complexity of the loss-transformed hypothesis class by $L \cdot \hat{\mathcal{R}}_S(H)$.
Let $\ell$ be a mean-reduced, coordinatewise $L$-Lipschitz loss (so that Talagrand--Ledoux
yields the factor $L/\sqrt{n}$). Combining Lemma~\ref{lem:mainresult},
Theorem~\ref{thm:dudley-hybrid}, and Talagrand--Ledoux’s contraction, we obtain for our dataset $X$:
$$
\begin{aligned}
&\hat{\mathcal R}_X(\ell\circ H)
\;\le\; \\
&\frac{L}{\sqrt{n}}\ \hat{\mathcal R}_X(H)
\;\in\;
\tilde{\mathcal O}\!\left(
\frac{L^{k}\alpha^{\,k}\beta}{\sqrt{N}}\,
\big( kR\sqrt{k m n}\;+\;\sqrt{T\log T}\big)
\right).
\end{aligned}
$$
To conclude the proof, we use Thm.~\ref{thm:Radbound}:
$\operatorname{gen}(h)\le 2\,\hat{\mathcal R}_X(\ell\!\circ\! H)+c\sqrt{\log(1/\delta)/N}$.
Talagrand--Ledoux’s contraction gives
$\hat{\mathcal R}_X(\ell\!\circ\! H)\le \frac{L}{\sqrt{n}}\hat{\mathcal R}_X(H)$.
Substitute the bound from Theorem~\ref{thm:dudley-hybrid}:
$\hat{\mathcal R}_X(H)\in \tilde{\mathcal O}\!\big(\frac{L^{\,k-1}\alpha^{\,k}\beta\sqrt{n}}{\sqrt{N}}
( kR\sqrt{k m n}+\sqrt{T\log T})\big)$,
which yields the stated result after canceling $\sqrt{n}$.
\end{proof}

\end{document}